\documentclass{jssc-submit}

\Volume{202X}{XX()}

\def\textsubscript#1%
{$_{\text{#1}}$}
\def\cdd{\mbox{\boldmath$\cdot$}~}

\makeatletter
\renewcommand{\fnum@table}{{\bfseries Table \thetable}}

\makeatother
\usepackage{epstopdf}
\usepackage{graphicx}
\usepackage{multirow}
\usepackage{mathtools}
\usepackage{extarrows}
\usepackage{amsmath,amssymb,amsfonts}
\usepackage{amsthm}
\usepackage{mathrsfs}
\usepackage[title]{appendix}
\usepackage{xcolor}
\usepackage{textcomp}
\usepackage{manyfoot}
\usepackage{booktabs}
\usepackage{algorithm}
\usepackage{algorithmicx}
\usepackage{algpseudocode}
\usepackage{listings}
\usepackage{caption}
\usepackage{float}
\usepackage[normalem]{ulem}
\usepackage{subcaption}
\usepackage[shortlabels]{enumitem}
\usepackage{hyperref}
\usepackage{indentfirst}
\usepackage[utf8]{inputenc}
\usepackage{xspace}
\usepackage{array}
\usepackage{atbegshi}
\usepackage{calc}

\shorttitle{DEGENERACY OF SMALL ZERO-ONE NETWORKS}
\shortauthor{X.~Tang, Y.~Wang, and J.~Zhang}

\newcommand{\defword}{\textit}

\theoremstyle{plain}

\theoremstyle{definition}

\newtheorem{newalgorithm}[theorem]{Algorithm}
\theoremstyle{remark}

\begin{document}
\thispagestyle{first}

\title{Degeneracy of Two-Dimensional Zero-one Reaction Networks with Up to Three Species}
{\uppercase{Tang} Xiaoxian \cdd \uppercase{Wang} Yihan \cdd \uppercase{Zhang} Jiandong}
{\uppercase{Tang} Xiaoxian\\
School of Mathematical Sciences, Beihang University, Beijing $100191$, China. Email: xiaoxian@buaa.edu.cn.\\
\uppercase{Wang} Yihan \cdd \uppercase{Zhang} Jiandong ({\bf Corresponding author})\\
School of Mathematical Sciences, Beihang University, Beijing $100191$, China. Email: 22091013wyh@buaa.edu.cn; zhangjiandong@buaa.edu.cn.}
{{$^\diamond${\it This paper was recommended for publication by Editor . }}}

\drd{DOI: }{Received: x x 20xx}{ / Revised: x x 20xx}

\Abstract{Zero-one biochemical reaction networks are widely recognized for their importance in analyzing signal transduction and cellular decision-making processes. 
Degenerate networks reveal non-standard  behaviors and mark the boundary where classical methods fail. Their analysis is key to understanding exceptional dynamical phenomena in biochemical systems. Therefore, we focus on investigating the degeneracy  of zero-one reaction networks. It is known that one-dimensional zero-one networks cannot degenerate. In this work,  we identify all degenerate two-dimensional zero-one reaction networks with up to three species by an efficient algorithm. By analyzing the structure of these networks, we arrive at the following conclusion: if a two-dimensional zero-one reaction network with three species is degenerate, then its steady-state system is equivalent to a binomial system.}

\section{Introduction}\label{sec1}

For the dynamical systems that arise from biochemical reaction networks, we ask the following basic question. 
{
\begin{question}\label{question}
Which reaction network is degenerate?
\end{question}
}
First, we explain degeneracy by the following network:
\begin{align*}
&{X_1+X_3 \xrightarrow{\kappa_1} 0},
&&{X_2 \xrightarrow{\kappa_2}X_1+ X_3},
&&{X_1  + X_3 \xrightarrow{\kappa_3} X_1 + X_2 + X_3}.
\end{align*}
%It has the following steady-state systems (see how to write down the system for a given network later in (\ref{eq:sys}))
Let $x_i$ denote the concentration of species $X_i$ $(i\in \{{1,2,3}\})$. Under the mass-action assumption, the time evolution of $x_i$ is governed by the following ordinary differential equation (ODE) system
\begin{align*}
\left\{
\begin{aligned}
    \dot{x}_1 &= -\kappa_1x_1x_3 + \kappa_2x_2, \\
    \dot{x}_2 &= -\kappa_2x_2+ \kappa_3 x_1x_3 , \\
    \dot{x}_3 & = -\kappa_1x_1x_3 + \kappa_2x_2 .
\end{aligned}
\right.
\end{align*}
 
The stoichiometric matrix is
\[
\mathcal{N} =\left( \begin{array}{rrr}
-1 & 1 & 0 \\
0 & -1 & 1 \\
-1 & 1 & 0
\end{array}
\right),
\]
which has rank $2$, so this network is two-dimensional. The conservation law is \(x_1 - x_3 = c\), where $c\in {\mathbb R}$. So,
the augmented steady-state system \(h\) (with conservation law) is
\begin{align*}
\left\{
\begin{aligned}
    h_1 &= -\kappa_1x_1x_3 + \kappa_2x_2, \\
    h_2 &= -\kappa_2x_2+ \kappa_3 x_1x_3, \\
    h_3 &= x_1 - x_3 - c. \\
\end{aligned}
\right.
\end{align*}
 By $h_1=h_2=0$,  we obtain the equations 
$$x_2/x_1x_3 = \kappa_1/\kappa_2= \kappa_3/\kappa_2.$$ 
So, if $\kappa_1 = \kappa_3$, the system admits infinitely many positive steady states satisfying 
$x_2/x_1x_3 = \kappa_1/\kappa_2$ together with $x_1 - x_3 = c$.  And if $\kappa_1\neq \kappa_3$, then there is no positive steady states. 
Note that for the Jacobian matrix of $h$ with respect to $x$, we have 
\[
\det(\operatorname{Jac}_h) =
\begin{array}{|rrr|}
-\kappa_1x_3 & \kappa_2 & -\kappa_1 x_1 \\
\kappa_3x_3 & -\kappa_2 & \kappa_3 x_1 \\
1 & 0 & -1  \\
\end{array}
=\kappa_2(\kappa_3-\kappa_1)(x_1+x_3).
\]
Hence, we observe that 
 $\det(\operatorname{Jac}_h)$ is identically zero at any positive steady state.

Thus, the network admits only degenerate positive steady states. In this case, we say this network is degenerate.
Reversely, a network is said to be nondegenerate if it admits at least one nondegenerate positive steady state. See the formal definitions in Section \ref{sec:back}.

Studying the degeneracy or the nondegeneracy of reaction networks is essential for understanding diverse dynamical behaviors in biochemical systems. Nondegeneracy, along with other key properties such as multistability, Hopf bifurcations, and absolute concentration robustness (ACR), underlies switching behavior, oscillations, and cellular decision-making in signaling systems \cite{ conradi2019existence, 2208.04196, BanajiBoros, Bihan2020, Mueller2016, Puente2025}.
While nondegenerate networks exhibit well-behaved steady state structures that are robust under parameter perturbations, degenerate networks can display unexpected geometric features, such as positive steady state sets of higher dimension or singularities in the Jacobian (see related discussion in \cite[Theorem 3.1]{Feliu2024}). These phenomena may reflect critical or pathological regimes in biological systems. Moreover, degenerate networks help characterize boundary cases where standard tools from algebraic geometry or dynamical systems theory break down. For these reasons, we ask the following questions: (a) how to efficiently determine whether a given network is degenerate; (b) what special structural properties degenerate networks possess.

Since studying large biochemical reaction networks is challenging, {some recent works have attempted to address this difficulty by exploiting structural properties of large biological systems, such as sparsity, chordal graphs, and block triangular structure \cite{CCW, MJ}. In} this work, we study small reaction networks, motivated by the observation that many important dynamical behaviors, such as multistability \cite{BanajiPantea2018, JoshiShiu2013}, oscillations \cite{Banaji2018}, and local bifurcations \cite{BanajiBorosHofbauer2023}, can be inherited from large networks to  smaller subnetworks. In recent studies, considerable effort has been devoted to identifying the minimal networks within broad classes that can exhibit these complex dynamical features.
 Earlier, a series of studies characterized small networks capable of exhibiting bistability by performing computations on specific classes of networks with fixed numbers of species and reactions (e.g., \cite{Ramakrishnan2008,Wilhelm2009}). 
Recently, Tang and Xu \cite{TangXu2021} classified all minimal multistable reaction networks with two reactions, identifying exactly which small networks, under constraints on species and reactants, are capable of multistability. Banaji and Boros \cite{BanajiBoros} recently classified all minimal at-most-bimolecular networks that can exhibit Hopf bifurcations, showing that such networks necessarily consist of three species and four reactions. Tang and Wang \cite{2208.04196} identified the smallest zero-one networks capable of Hopf bifurcations as four-dimensional systems with four species and five reactions. Moreover, Kaihnsa, Nguyen, and Shiu \cite{KaihnsaNguyenShiu2024} established that any at-most-bimolecular network exhibiting both multistationarity and absolute concentration robustness (ACR) must have at least three species, three reactions, and a dimension of at least two. Jiao and Tang \cite{JiaoTang2025} developed an efficient algorithm for determining the equivalence of small zero-one reaction networks based on their steady-state ideals, enabling the classification of millions of networks while avoiding expensive Gröbner basis computations.

When the stoichiometric coefficients are limited to zero or one, the network is called a zero-one network.

For example, reactions such as $0 \to X$ and $X \to 0$ naturally appear in biochemical models: $0 \to X$ can represent the input of a substance (e.g., glucose) from an external source into a cell, while $X \to 0$ can represent the degradation of $X$ or its loss to the external environment.

Our interest in these networks stems from their prevalence in cell signaling, where many key biochemical systems exhibit this structure. Examples include the two-layer MAPK cascade \cite{Dickenstein2023, Zumsande2010}, hybrid histidine kinase systems \cite{JaniakSpens2005, Kothamachu2015}, and the ERK network \cite{Georgiev2006, Baudier2018}. 
A more comprehensive list of such networks from signaling pathways is provided in \cite[Figure 2]{TelekFeliu2023}, which presents eleven representative zero-one models.
%Among these, degenerate steady states, which occur when the Jacobian matrix restricted to the stoichiometric subspace is not surjective, play a particularly important role in dynamics of a system. Unlike nondegenerate steady states, degeneracy is often associated with complex or abnormal behavior, such as the presence of connected steady-state branches or bifurcation points.
%This motivates us to focus on the study of degenerate networks, particularly in low-dimensional settings where complete classification is feasible. Two-dimensional zero-one reaction networks with three species are precisely such a class of networks.
%These networks are structurally simple, allowing for exhaustive enumeration and symbolic analysis. We therefore carry out a computational investigation of this class of networks.

In this work, we focus on the degeneracy of small zero-one networks. It is known that a one-dimensional zero-one network 
either admits no positive steady states, or admits a unique nondegenerate positive steady state \cite[Theorem 2]{JTZ}. By a known  result \cite[Theorem 3]{JTZ}, any two-dimensional zero-one network involving no more than three species is incapable of supporting nondegenerate multistationarity, and if it has multiple positive steady states, they are necessarily degenerate. Here, we provide an efficient algorithm 
(Algorithm \ref{algorithm}) for determining whether a given reaction network admits only degenerate positive steady states (i.e., whether the network is degenerate). 
The correctness of the algorithm 
is based on the proof of Theorem \ref{thm:alg}. 
The efficiency of our method stems from a transformation of the Jacobian matrix based on extreme rays of the flux cone, which avoids direct symbolic representations of steady states and greatly reduces computational costs. The core of the algorithm involves checking whether the principal  minors of the partial transformed Jacobian matrix $A(\lambda)$
are the zero polynomial, which corresponds to the Jacobian matrix of the system augmented with the conservation laws failing to have full rank at all positive steady states. Our main contribution is applying the algorithm to all two-dimensional three-species zero-one reaction networks
and successfully identifying 3152 degenerate networks among more than a million candidates. The computational results summarized in Table \ref{tab:degenerate} show that most degenerate networks have certain types of conservation laws, where one species only depends on at most one of the other two species. The computational timings recorded in Table \ref{tab:timing-alg1} support the efficiency of Algorithm \ref{algorithm}. A key new finding beyond \cite{JTZ} is that, for all the 3152 degenerate networks, their steady-state systems are  equivalent to a binomial system. This empirical pattern is formalized in Theorem \ref{eq:thm1}.  
   Our computation also confirms that 
    all these degenerate networks admit infinitely many positive steady states, which implies that if a two-dimensional zero-one network with three species admits only one positive steady state, then it is nondegenerate, and these results are consistent with \cite{JTZ}. 

    Binomial systems are interesting in the history of studying reaction networks by algebraic methods; see \cite{PerezMillan2012, Sadeghimanesh2019, Rahkooy2021}.  In the study of biochemical reaction networks, networks with binomial or toric steady states occupy a central position due to their distinct algebraic structures and dynamical properties. First, the clear algebraic form of such networks facilitates computational analysis. Second, they serve as powerful hubs for analyzing complex dynamics. For many intricate networks involving intermediates, such as phosphorylation cycles, global dynamical properties such as multistationarity can essentially be determined and explained by a reduced binomial core networks. More importantly, from a biological perspective, binomial networks are not exceptional cases. Mathematical models of many key biological modules are found to naturally possess or approximate toric steady states, confirming the prevalence of these structures in real systems. Our discovery that all degenerate networks are binomial systems provides novel evidence for this universal and profound network architecture.

The rest of this paper is organized as follows. In Section \ref{sec:back}, we review the standard concepts in reaction networks, including the definitions of zero-one networks, steady states and degeneracy. In Section \ref{sec3},  we discuss the transformed Jacobian matrix and its equivalence to the original Jacobian matrix at steady states, providing the theoretical foundation for our degeneracy criterion in Theorem \ref{thm:alg}.
In Section \ref{secfour}, we formally propose Algorithm \ref{algorithm} for detecting degeneracy, along with Algorithm \ref{algorithm2} for enumerating maximum networks and Algorithm \ref{algorithm3} for preprocessing.
In Section \ref{sec4},  we implement these algorithms to obtain all degenerate two-dimensional three-species zero-one reaction networks, and we analyze the computational results.  
In Section \ref{sec6}, we present our main results, establishing Theorem \ref{eq:thm1} concerning the structural features of these networks based on computational analysis.
 Finally, in Section \ref{sec7}, we discuss the broader implications of our results and suggest directions for future research.

\section{Background}\label{sec:back}

A \defword{reaction network} $G$, or simply a \defword{network}, consists of $s$ species $\{X_1, X_2, \dots, X_s\}$ and $m$ reactions of the form 
\begin{align}\label{eq:network}
\mu_{1j}X_1 + \mu_{2j}X_2 +
 \dots +
\mu_{sj}X_s
~ \xrightarrow{\kappa_j} ~
\nu_{1j}X_1 + \nu_{2j}X_2 +
 \dots +
\nu_{sj}X_s,
 \quad
    \text{for } j=1,2, \ldots, m,
\end{align}
where all \defword{stoichiometric coefficients} $\mu_{ij}$ and $\nu_{ij}$ are nonnegative integers, and we assume that $(\mu_{1j},\mu_{2j},\ldots,\mu_{sj}) \ne (\nu_{1j},\nu_{2j},\ldots,\nu_{sj})$ for each reaction. Each $\kappa_j \in \mathbb{R}_{>0}$ denotes the \defword{rate constant} of the $j$-th reaction in \eqref{eq:network}. A reaction is called a \defword{zero-one reaction} if all its stoichiometric coefficients are either 0 or 1. A \defword{zero-one network} consists solely of zero-one reactions. For each reaction, we define the \defword{stoichiometric vector} as
\begin{align}\label{eq:vector}
\Delta_j := (\nu_{1j} - \mu_{1j}, \nu_{2j} - \mu_{2j}, \ldots, \nu_{sj} - \mu_{sj})^\top.
\end{align}
The \defword{stoichiometric matrix} $\mathcal{N}$ of $G$ is an $s \times m$ matrix, where the $(i, j)$-entry of $\mathcal{N}$ is defined as $\nu_{ij} - \mu_{ij}$. The \defword{reactant matrix} $\mathcal{X}$ of $G$ is an $s \times m$ matrix, where the $(i, j)$-entry of $\mathcal{X}$ is defined as $\mu_{ij}$. The real linear space spanned by the column vectors $\Delta_1, \Delta_2,\ldots, \Delta_m$ of $\mathcal{N}$ defines the \defword{stoichiometric subspace}, denoted by $S$.

Let $x_1, x_2,\ldots, x_s$ denote the concentrations of species $X_1,X_2, \ldots, X_s$. Under the mass-action assumption, their time evolution is governed by the following ODE system
\begin{align}\label{eq:sys}
\dot{x} = f(\kappa, x) := \mathcal{N}v(\kappa, x),
\end{align}
where $x := (x_1, x_2, \ldots, x_s)^\top$, and $v(\kappa, x) := (v_1(\kappa, x),v_2(\kappa, x), \dots, v_m(\kappa, x))^\top$ with
\begin{align}\label{eq:v}
v_j(\kappa, x) := \kappa_j \prod_{i=1}^s x_i^{\mu_{ij}}. 
\end{align}
Treating $\kappa := (\kappa_1,\kappa_2, \ldots, \kappa_m)^\top$ as a  vector of parameters, we have  $f_{i}(\kappa,x) \in \mathbb{Q}(\kappa)[x]$, for $i\in\{1,2,\dots, s\}$. 

Let $d := s - \text{rank}(\mathcal{N})$. A \defword{conservation-law matrix} $W$ is a $d \times s$ matrix with row-reduced echelon form whose rows form a basis for the orthogonal complement $S^\perp$ of the stoichiometric subspace. Note that $\text{rank}(W) = d$, and system \eqref{eq:sys} satisfies $W \dot{x} = \mathbf{0}$, where $\mathbf{0}$ denotes the vector whose coordinates are all zero. Thus, any solution $x(t)$ with nonnegative initial condition $x(0) \in \mathbb{R}_{\ge 0}^s$ remains within the \defword{stoichiometric compatibility class}
\begin{align}\label{eq:pc}
\mathcal{P}_c := \{x \in \mathbb{R}_{\geq 0}^s \mid Wx = c\}, \quad c := Wx(0) \in \mathbb{R}^d.
\end{align}
The \defword{positive stoichiometric compatibility class} is the relative interior of $\mathcal{P}_c$
\begin{align*}
\mathcal{P}_c^+ := \{x \in \mathbb{R}_{> 0}^s \mid Wx = c\} = \mathcal{P}_c \cap \mathbb{R}_{>0}^s.
\end{align*}
Let $I$ denote the set $\{i_1, i_2, \dots, i_d\}$, where \(i_j\) is the column index of the first nonzero entry in the \(i\)-th row of \(W\), and we let $i_1 < i_2 < \cdots < i_d$. Define
\begin{equation}
h_i := 
\begin{cases}
f_i & \text{if } i \notin I, \\
(Wx - c)_k & \text{if } i = i_k \in I,
\end{cases} 
\label{h_system}
\end{equation}
where $f_1, f_2,\dots, f_s$ are defined in \eqref{eq:sys}.  
Then we define
\begin{align}\label{eq:h}
h := (h_1, h_2,\ldots, h_s),
\end{align}
and we refer to system \eqref{eq:h} as the \defword{steady-state system augmented by conservation laws}.

Given a rate-constant vector $\kappa^* \in \mathbb{R}_{>0}^m$, a \defword{steady state} of system~\eqref{eq:sys} refers to a concentration vector $x^* \in \mathbb{R}_{\geq 0}^s$ such that  $f(\kappa^*, x^*) = \mathbf{0}$, where $f(\kappa, x)$ represents the right-hand side of the ODE system~\eqref{eq:sys}. If every entry of $x^*$ is strictly positive, that is $x^* \in \mathbb{R}_{> 0}^s$, then $x^*$ is referred to as a \defword{positive steady state}. We say a steady state \( x^* \) is \defword{degenerate} if the Jacobian matrix \( \text{Jac}_f(\kappa^*, x^*) \) restricted to the \defword{stoichiometric subspace} \( S \) is not surjective, i.e.,  $
\text{im}\left(\text{Jac}_f(\kappa^*, x^*)|_S \right) \neq S.$
%Equivalently, if the \defword{stoichiometric subspace} \( S \) has dimension \( r \), then \( x^* \) is \defword{degenerate} if and only if $\text{rank}\left(\text{Jac}_f(\kappa^*, x^*)|_S \right) < r$. 
For any $\kappa^* \in \mathbb{R}_{>0}^m$ and for any $c^* \in \mathbb{R}^d$, a solution $x^* \in \mathbb{R}_{\ge 0}^s$ of $h = \mathbf{0}$ is said to be \defword{a steady state  in ${\mathcal P}_{c^*}$}. Notice that a steady state $x^*$ is \defword{degenerate} if and only if the Jacobian matrix $\text{Jac}_h(\kappa^*, x^*)$ does not have full rank \cite[Definition 9.9]{ConradiPantea_multistationarity}. We remark that if $s=\text{rank}(\mathcal{N})$, then according to the definition of $h$ \eqref{h_system}, we have $h=f$ (i.e., there are no conservation laws), and $x^*$ is degenerate if and only if $\det(\text{Jac}_f(\kappa^*, x^*))=0$. However, in most cases, we have $s>\text{rank}(\mathcal{N})$. Then, the rows of $\mathcal{N}$ are linearly dependent, and by \eqref{eq:sys}, we know that $\det(\text{Jac}_f)$ is identically zero. Hence, in this case, one cannot conclude degeneracy by checking $\det(\text{Jac}_f)$.  If a network $G$ admits only degenerate positive steady states, then it is called a \defword{degenerate network}. If a network $G$ admits at least one nondegenerate positive steady state, then it is called a \defword{nondegenerate network}.  Here, we define the degeneracy/nondegeneracy of a network following from a recent study \cite{Feliu2024}.  Notice that if a network is not degenerate, then it is nondegenerate.

For the convenience of subsequent algorithmic and classification descriptions, we introduce several useful concepts.
Let $G$ be an $r$-dimensional zero-one reaction network with $s$ species.
The network $G$ is called a \defword{maximum $s$-species network} if adding any additional zero-one reaction involving only the species $X_1, X_2,\dots, X_s$ strictly increases the dimension of the network to $r+1$ \cite[Definition~2]{JTZ}.
Two networks are said to \textit{have the same form} if one can be obtained from the other by relabeling the species $X_1, X_2,\ldots, X_s$ or the reactions $\mathcal{R}_1, \mathcal{R}_2,\ldots, \mathcal{R}_m$ \cite[Definition~1]{JiaoTang2025}.
A \textit{subnetwork} of $G$ is a network obtained by selecting a subset of reactions from $G$ that has the same dimension as $G$; in particular, $G$ itself is considered a subnetwork \cite[Definition~2.1]{joshi2017small}.
A species is called \defword{trivial} if, in every reaction, it either does not appear or appears on both sides of the reaction \cite[Section~2.1]{BBH2024}.
Two reaction networks are said to be \defword{equivalent} if one can be obtained from the other by removing all trivial species.  We remark that it is common that one species appears on both sides of a certain reaction (e.g., an enzyme), but usually such a species is not necessarily trivial in most well-known biochemical reaction networks such as the ERK network. For instance, the species $X_1$ in the network
$$X_1\xrightarrow{{\kappa_1}} X_1 + X_2, \;\; 0\xrightarrow{{\kappa_2}}  X_3$$
is trivial, while it is not trivial in the network
$$X_1\xrightarrow{{\kappa_1}} X_1 + X_2, \;\; X_1\xrightarrow{{\kappa_2}}  X_3.$$

% \begin{definition}\cite[Definition 2]{JTZ}}
% \label{def:maximum network}
% Let $G$ be an $r$-dimensional zero-one reaction network with $s$ species. We say that $G$ is a \defword{maximum $s$-species network} if for all} additional zero-one reaction involving only the species $X_1, \dots, X_s$, the dimension of the resulting network increases strictly to $r+1$.
% \end{definition}

% \begin{definition}\cite[Definition 1]{JiaoTang2025}}
% \label{def:the same form networks 1}
% We say that a network $G'$ \textit{has the same form with} another network $G$ if $G'$ can be derived from $G$ through a relabeling of the species $X_1, \ldots, X_s$ as $X'_1, \ldots, X'_s$, or a relabeling of the reactions $\mathcal{R}_1, \ldots, \mathcal{R}_m$ as $\mathcal{R}'_1, \ldots, \mathcal{R}'_m$.
% \end{definition}

% \begin{definition}\cite[Definition 2.1]{joshi2017small}}
% \label{eq:subnetwork}
% For any network $G$, a \textit{subnetwork} of $G$ is a network composed of some reactions in $G$ and has the same dimension as $G$}. Note that $G$ itself is also  a subnetwork of $G$.
% \end{definition}

% \begin{definition}\cite[Section 2.1]{BBH2024}}
%     A species in a reaction network is called \defword{redundant} if in each reaction, the species either does not appear  or appears on both sides of the reaction (i.e., the species is both a reactant and a product). 
% \end{definition}

% \begin{definition}
%     Two reaction networks are said to be \defword{equivalent} if one of them can be obtained from the other one by removing all redundant species.
% \end{definition}

\section{Methods}\label{sec3}

In this section,  we  first recall a transformation of the Jacobian matrix based on extreme rays. In Lemma \ref{lemma2}, we present the relationship between  the Jacobian matrices before and after transformation. Then, we present an algorithm (Algorithm \ref{algorithm}) to determine whether a network is degenerate, and we prove the correctness of the algorithm in Theorem \ref{thm:alg}. 

For a  reaction network \( G \) with \( s \) species and \( m \) reactions, let \( \mathcal{N} \in \mathbb{R}^{s \times m} \) be the stoichiometric matrix, and let \( \mathcal{X} \in \mathbb{R}^{s \times m} \) be the reactant matrix.  As a standard approach  (see, e.g., \cite{conradi2019existence}), the Jacobian matrix \( \operatorname{Jac}_f(\kappa, x) \in \mathbb{R}^{s \times s} \) associated with  \(  f(\kappa, x) \) defined  in \eqref{eq:sys} can be written as 
\begin{align}\label{eq:jacf}
\operatorname{Jac}_f(\kappa, x) = \mathcal{N} \operatorname{diag}\big(v(\kappa, x)\big) \mathcal{X}^\top \operatorname{diag}(p),
\end{align}
where
\begin{itemize}
    \item[(i)] $ \operatorname{diag}\big(v(\kappa, x)\big) $ is an $ m \times m $ diagonal matrix with $ v(\kappa, x) $ defined in \eqref{eq:v} on its diagonal, 
    
    \item[(ii)] $ p := (p_1,p_2, \ldots, p_s)^\top = \left( \frac{1}{x_1}, \frac{1}{x_2},\ldots, \frac{1}{x_s} \right)^\top $, and $ \operatorname{diag}(p) $ is an $ s \times s $ diagonal matrix with entries $ p_i $ on its diagonal.
\end{itemize}
Next, we analyze the Jacobian matrix $\text{Jac}_f(\kappa,x)$ under a coordinate transformation evaluated at positive steady states. %\begin{definition}%[Flux cone]
For the stoichiometric matrix $\mathcal{N} \in \mathbb{R}^{s \times m}$, the corresponding \defword{flux cone} is defined by 
\begin{align}\label{eq:Fn}
\mathcal{F}(\mathcal{N}) := \left\{ \alpha \in \mathbb{R}_{\geq 0}^m \mid \mathcal{N} \alpha = \mathbf{0} \right\},
\end{align}
%\end{definition}
%\noindent The following lemma is a direct consequence of the Minkowski-Weyl theorem \cite{conradi2019existence}.
%\begin{lemma}%[Extreme ray representation]
which is a convex polyhedral cone generated by a basis, i.e., a finite set of extreme rays (see, e.g., \cite{2208.04196}). By representing any \defword{flux vector} \(\alpha \in \mathcal{F}(\mathcal{N})\) as a nonnegative combination of these extreme rays, we can reparameterize the Jacobian matrix in terms of convex parameters \((\lambda, p)\) as follows.
Let $l^{(1)},l^{(2)},\ldots,l^{(t)}\in\mathbb{R}_{\geq0}^m$ be a set of extreme rays of $\mathcal{F}(\mathcal{N})$. Then, any $\alpha\in \mathcal{F}(\mathcal{N})$ can be written as
\begin{align}\label{eq:alpha}
\alpha = \sum_{i=1}^{t} \lambda_i l^{(i)}, \quad \text{with } \lambda_i \geq 0 \ \text{for any } i \in \{1, 2,\ldots, t\},
\end{align}
and we define $ \lambda := (\lambda_1, \lambda_2,\ldots, \lambda_t)^\top $. 
The partial transformed Jacobian matrix in terms of $\lambda$ is defined as

\begin{equation} \label{equation11}
    A(\lambda) := \mathcal{N} \operatorname{diag}\left(\sum_{i=1}^t \lambda_i l^{(i)}\right) \mathcal{X}^\top.
\end{equation}

The transformed Jacobian matrix in terms of $(p, \lambda)$ is defined as
\begin{align}\label{eq:jhp}
J(p, \lambda) := \mathcal{N} \ \text{diag}\left( \sum_{i=1}^{t} \lambda_i l^{(i)} \right) \mathcal{X}^{\top} \ \text{diag}(p).
\end{align}

For all $I \subseteq \{1,2, \ldots, s\}$, we have
\begin{equation} \label{eq:3.13}
    \det\big(J(p, \lambda)[I, I]\big) = \det\big(A(\lambda)[I, I]\big) \cdot \prod_{j \in I} p_j.
\end{equation}

\begin{lemma}
\cite[Lemmas 4.1 and 4.3]{2208.04196}
\label{lemma2}
Let \( G \) be a network as defined in \eqref{eq:network}, and let \( f \) denote the steady-state system given by \eqref{eq:sys}. 
    Suppose \( J(p, \lambda) \in \mathbb{Q}[p, \lambda]^{s \times s} \) is the matrix defined in \eqref{eq:jhp}. 
    Then, for any \( \kappa \in \mathbb{R}_{>0}^m \) and for any associated positive steady state \( x \in \mathbb{R}_{>0}^s \), 
    there exist  \( p \in \mathbb{R}_{>0}^s \) and \( \lambda \in \mathbb{R}_{\geq0}^t \) satisfying 
    \[
         \operatorname{Jac}_f(\kappa, x)=J(p, \lambda),
    \]
    and for any  \( p \in \mathbb{R}_{>0}^s \) and \( \lambda \in \mathbb{R}_{\geq0}^t \),
    there exist  \( \kappa \in \mathbb{R}_{>0}^m \) and  associated positive steady state \( x \in \mathbb{R}_{>0}^s \) satisfying 
    \[
        J(p, \lambda)= \operatorname{Jac}_f(\kappa, x).
    \]
\end{lemma}

\begin{lemma}\cite[Proposition 5.3]{WiufFeliu_powerlaw}\label{lemma3}
Let \( M \in \mathbb{R}^{s \times s} \) be a real matrix, and let \( r \leq s \) be an integer.  
Suppose \( F \subseteq \mathbb{R}^{{s}}\) is an \( r \)-dimensional vector space that contains the space generated by the columns of \( M \).  
Let \( F^\perp \) be the orthogonal complement of \( F \), and its dimension is \( d:=s - r \). 
%Let \( \{\omega_1, \dots, \omega_d\} \subseteq \mathbb{R}^s \) be a \emph{reduced basis} of \( F^\perp \), meaning that each \( \omega_i \) satisfies \( (\omega_i)_i = 1 \) and \( (\omega_i)_j = 0 \) for all \( j < i \).  
Let a new matrix \( \widetilde{M} \in \mathbb{R}^{s \times s} \) whose first \( d \) rows are \( \omega_1, \omega_2,\dots, \omega_d \), and whose last \( r \) rows are the corresponding last \( r \) rows of \( M \).
Then:
\[
\det(\widetilde{M}) = \sum_{\substack{I \subseteq \{1,2, \dots, s\} \\ |I| = r}} \det(M[I,I]),
\]
where \( M[I,I] \) denotes the \( r \times r \) principal submatrix of \( M \) with row and column indices in \( I \).
Moreover, if \( \operatorname{rank}(M) < r \), then both sides of the equation are zero.

\end{lemma}

\begin{theorem}\label{thm:alg}

Let $G$ be an $r$-dimensional reaction network with $s$ species, and let $A(\lambda)$ be the partial transformed Jacobian matrix defined in \eqref{equation11} above. Then the network $G$ is \defword{degenerate} if and only if
$
\det\bigl(A(\lambda)[I, I]\bigr) \text{ is the zero polynomial for all }  I \subseteq \{1,2,\dots,s\} \text{ with } |I| = r.
$

\end{theorem}
\begin{proof}
Let $f$ be the steady-state system defined as in \eqref{eq:sys}. Let $h$ be the \defword{steady-state system augmented by conservation laws}, as defined in \eqref{eq:h}. Let \( J(p, \lambda) \in \mathbb{Q}[p, \lambda]^{s \times s} \) be the transformed Jacobian matrix  as defined in \eqref{eq:jhp}. 
Since $G$ is an $r$-dimensional network with $s$ species, it follows from Lemma~\ref{lemma3} that
\begin{align}\label{9}
\det(\operatorname{Jac}_h) = \sum_{\substack{I \subseteq \{1,2, \ldots, s\} \\ |I| = r}}  \det(\operatorname{Jac}_f[I, I]).
\end{align}

We define the polynomial 
$$
B(p, \lambda) := \sum_{\substack{I \subseteq \{1,2, \ldots, s\} \\ |I| = r}} \det\left( J(p, \lambda)[I, I] \right).
$$
Substituting the above equation \eqref{eq:3.13} yields:
\[
B(p, \lambda) = \sum_{\substack{I \subseteq \{1,2,\dots,s\} \\ |I| = r}} \Bigl( \det\bigl(A[I, I]\bigr) \cdot \prod_{j \in I} p_j \Bigr).
\]
Observe that in the expression above, the factors $\prod_{j \in I} p_j$ are independent positive monomials in the variables $p$. Therefore, the polynomial $B(p, \lambda)$ is identically zero if and only if 
$
\det\bigl(A(\lambda)[I, I]\bigr) \text{ is the zero polynomial for all }  I \subseteq \{1,2,\dots,s\} \text{ with } |I| = r.
$

By equation \eqref{9} and by Lemma \ref{lemma2}, for any $\kappa \in \mathbb{R}_{>0}^m$ and any positive steady state $x \in \mathbb{R}_{>0}^s$, there exist \( p \in \mathbb{R}_{>0}^s \) and \( \lambda \in \mathbb{R}_{\geq0}^t \) such that $\det(\text{Jac}_h(\kappa, x)) = B(p, \lambda)$. Since \( B(p, \lambda) \equiv 0 \), it follows that $\det(\text{Jac}_h(\kappa, x))=0$. Therefore, the network $G$ admits only degenerate positive steady states. According to the definition of  \defword{degenerate network}, the network \(G\) is degenerate. On the other hand,  when ${\mathcal F}(\mathcal{N})\neq \emptyset$, 
if $B(p, \lambda)$ is not the zero polynomial, then there always exist $p\in \mathbb{R}_{>0}^s$ and $\lambda \in \mathbb{R}_{\geq0}^t$ such that $B(p, \lambda)\neq 0$, and so, by equation \eqref{9} and by Lemma \ref{lemma2}, the network admits at least one nondegenerate positive steady state. Therefore, the network $G$ is nondegenerate.
\end{proof}

% 
% \begin{remark}
% This example illustrates that the core of Algorithm~1 is computing the determinants of the $r \times r$ principal minors of $A(\lambda)$. Since this involves only basic linear-algebraic operations on matrices with polynomial entries, the algorithm is efficient and avoids more expensive symbolic computations.
% \end{remark}
% }

\section{Algorithms}\label{secfour}

Algorithm~\ref{algorithm} determines whether a  reaction network $G$, which admits  at least one positive steady state,  is degenerate. The algorithm proceeds as follows.

\begin{newalgorithm}
\textit{DetermineDegeneracy}
\label{algorithm}

\begin{description}[leftmargin=!, labelwidth=\widthof{\bf Step 6.}, align=left]
  \item[\bf Input:] The stoichiometric matrix $\mathcal{N} \in \mathbb{Z}^{s \times m}$ and the reactant matrix $\mathcal{X} \in \mathbb{Z}^{s \times m}$ for an $r$-dimensional zero-one reaction network with $s$ species and $m$ reactions.
  
  \item[\bf Output:] If the input network is degenerate, return \texttt{True}. If it is nondegenerate, return \texttt{False}. 
  %If it does not admit positive steady states, return \texttt{Null}.
  
  %\item[\bf Step 1.] First, check the emptiness of ${\mathcal F}({\mathcal N})$. If ${\mathcal F}({\mathcal N})=\emptyset$, then the network does not admit positive steady states, return \texttt{Null}. Otherwise, proceed to the next step.

  \item[\bf Step 1.] Compute the extreme rays $l^{(1)}, l^{(2)},\ldots, l^{(t)} \in \mathbb{R}_{\ge 0}^m$ of the flux cone ${\mathcal F}({\mathcal N})$.

  \item[\bf Step 2.] Define
\[
\alpha = \sum_{i=1}^t \lambda_i \, l^{(i)}, \quad \text{where } \lambda_i \geq 0 \text{ for all } i \in \{1,2, \ldots, t\}.
\]

  \item[\bf Step 3.] Construct the partial transformed Jacobian matrix
\[
A(\lambda) := \mathcal{N} \cdot \operatorname{diag}(\alpha) \cdot \mathcal{X}^\top.
\]

  \item[\bf Step 4.] Check the following condition:
\[
\bigwedge_{\substack{I \subseteq \{1,2, \ldots, s\} \\ |I| = r}} \Big( \det\!\bigl(A(\lambda)[I, I]\bigr) \equiv 0 \Big).
\]
 If this condition holds, then the network is degenerate, return \texttt{True}. Otherwise, return \texttt{False}.

\end{description}

\end{newalgorithm}

\begin{remark}
We will later present a dedicated preprocessing procedure, Algorithm~\ref{algorithm3}, which checks whether a given network admits any positive steady state. Therefore, when we invoke Algorithm~\ref{algorithm}, we assume by default that every input network admits at least one positive steady state, which means in Step 1, we can indeed obtain the extreme rays. 
\end{remark}

\begin{remark}
The core of Algorithm~\ref{algorithm} is computing the determinants of the $r \times r$ principal minors of $A(\lambda)$. Since this involves only basic linear-algebraic operations on matrices with polynomial entries, the algorithm is efficient and avoids more expensive symbolic computations.
\end{remark}

\begin{example}
This example illustrates how Algorithm 1 works. All computational steps are provided below, and we now elaborate on them in detail. Consider the following network:
\[
0 \xrightarrow{\kappa_1} X_1 + X_2,\quad 
0 \xrightarrow{\kappa_2} X_3,\quad 
X_1 + X_2 + X_3 \xrightarrow{\kappa_3} 0.
\]
\noindent \textbf{Step 1.} Compute extreme rays of the flux cone \(\mathcal{F}(\mathcal{N})\).
 Notice that the stoichiometric matrix is
\[
\mathcal{N} = \begin{pmatrix}
1 & 0 & -1 \\
1 & 0 & -1 \\
0 & 1 & -1
\end{pmatrix}.
\]We look for \(\alpha = (\alpha_1,\alpha_2,\alpha_3)^\top \in \mathbb{R}^3_{\geq 0}\) such that \(\mathcal{N}\alpha = \mathbf{0}\), i.e., 
\[
\begin{cases}
\alpha_1 - \alpha_3 = 0, \\
\alpha_1 - \alpha_3 = 0, \\
\alpha_2 - \alpha_3 = 0.
\end{cases}
\]
Thus \(\alpha_1 = \alpha_2 = \alpha_3 \geq 0\). Choosing \(\alpha_3 = 1\) gives that  the flux cone is generated by a single extreme ray:
\[
l^{(1)} = (1,\;1,\;1)^\top.
\]
\noindent \textbf{Step 2.} Since the flux cone is generated by the extreme ray \(l^{(1)}\), any flux vector can be written as
\[
\alpha = \lambda_1 l^{(1)} = (\lambda_1,\;\lambda_1,\;\lambda_1)^\top, \quad \lambda_1 \geq 0.
\]

\noindent \textbf{Step 3.} Construct the  partial transformed Jacobian matrix 
\begin{align*}
A(\lambda) &= \mathcal{N} \cdot \operatorname{diag}(\alpha) \cdot \mathcal{X}^\top\\
&= \begin{pmatrix}
1 & 0 & -1 \\
1 & 0 & -1 \\
0 & 1 & -1
\end{pmatrix}
\begin{pmatrix}
\lambda_1 & 0 & 0 \\
0 & \lambda_1 & 0 \\
0 & 0 & \lambda_1
\end{pmatrix}
\begin{pmatrix}
0 & 0 & 1 \\
0 & 0 & 1 \\
0 & 0 & 1
\end{pmatrix}^\top\\
&=\begin{pmatrix} -\lambda_1 & -\lambda_1 & -\lambda_1 \\ -\lambda_1 & -\lambda_1 & -\lambda_1 \\ -\lambda_1 & -\lambda_1 & -\lambda_1 \end{pmatrix}.
\end{align*}
\noindent \textbf{Step 4.} 
For this two-dimensional network (\(r=2\)), we compute all \(2 \times 2\) principal minors of \(A(\lambda)\). Obviously, we have 
\begin{align*}
\det\!\bigl(A(\lambda)[\{1,2\},\{1,2\}]\bigr) &= \det\begin{pmatrix} -\lambda_1 & -\lambda_1 \\ -\lambda_1 & -\lambda_1 \end{pmatrix} = 0, \\
\det\!\bigl(A(\lambda)[\{1,3\},\{1,3\}]\bigr) &= \det\begin{pmatrix} -\lambda_1 & -\lambda_1 \\ -\lambda_1 & -\lambda_1 \end{pmatrix}  = 0, \\
\det\!\bigl(A(\lambda)[\{2,3\},\{2,3\}]\bigr) &= \det\begin{pmatrix} -\lambda_1 & -\lambda_1 \\ -\lambda_1 & -\lambda_1 \end{pmatrix} = 0.
\end{align*}
 Since all \(2 \times 2\) principal minors of \(A(\lambda)\) are identically zero (as polynomials in \(\lambda_1\)), the algorithm returns \texttt{True}, confirming that the network is degenerate.
\end{example}

Algorithm~\ref{algorithm2}
enumerates all maximum $r$-dimensional zero-one reaction networks with $s$ species. The algorithm proceeds as follows.

\begin{newalgorithm}
\textit{EnumeratingMaximumNetworks}
\label{algorithm2}

\begin{description}[leftmargin=!, labelwidth=\widthof{\bf Step 5.}, align=left, itemsep=0pt, parsep=0pt]
  \item[\bf Input:] The number of species $s$ and the dimension of the network $r$.
   
  \item[\bf Output:] All maximum $r$-dimensional zero-one reaction networks with $s$ species.
  
  \item[\bf Step 1.] Define the set of all possible stoichiometric vectors for zero-one reactions with $s$ species according to~\eqref{eq:vector}:
        \[
        V := \{-1,0,1\}^s \setminus \{(\underbrace{0,0,\dots,0}_{s})\} \subseteq \mathbb{R}^s,
        \]

  \item[\bf Step 2.] Enumerate all tuples of linearly independent vectors $(a_1, a_2, \dots, a_r)$ in $V^r$.

  \item[\bf Step 3.] Initialize a set \(\mathcal{M}\) to store matrices.

  \item[\bf Step 4.] For each tuple $(a_1, a_2, \dots, a_r)$ from Step 2, construct a maximal rank-$r$ matrix as follows.
\begin{enumerate}
          \item[{\bf 4.1.}] Initialize a matrix with columns $a_1$, $a_2$, \dots,  $a_r$.
          \item[{\bf 4.2.}] For each remaining vector $a \in \ V \setminus \{a_1, a_2, \dots, a_r\}$, perform
                \begin{enumerate}
                  \item[{\bf 4.2.1.}] Form a new matrix by appending $a$ as a new column to the current matrix.
                  \item[{\bf 4.2.2.}] If the rank of the new matrix is $r$, then keep $a$ as a new column; otherwise, discard $a$.
                \end{enumerate}
          \item[{\bf 4.3.}] The above steps yield a matrix $M$ whose columns include $a_1, a_2, \dots, a_r$. Add $M$ to the set $\mathcal{M}$.
        \end{enumerate}

  \item[\bf Step 5.] For each matrix in $\mathcal{M}$, convert it into the corresponding maximum networks.

  \item[\bf Step 6.] For each network obtained in Step~5, remove the networks that have the same form as it.
\end{description}
\end{newalgorithm}

\begin{remark}\label{rmk:stoi_to_reaction}
%After obtaining the resulting rank-2 matrices in Step~3 of Algorithm~\ref{algorithm2}, we enumerate all corresponding maximum networks in Step~4. 
Note that a stoichiometric vector containing zero elements may correspond to more than one zero-one reaction. For instance, the vector $(1,0,1)$ corresponds to the following two reactions
\begin{align*}
0 \xrightarrow{\kappa_1} X_1 + X_3, \quad X_2 \xrightarrow{\kappa_2} X_1 + X_2 + X_3.
\end{align*}
Therefore, to obtain a maximum two-dimensional zero-one network with three-species, we must include all such reactions that yield the corresponding stoichiometric vector.
\end{remark}

%\subsection{Preprocessing}
%\begin{definition}\cite[Definition 2.1]{joshi2017small}}
%\label{eq:subnetwork_1}
%For any network $G$, a \textit{subnetwork} of $G$ is a network composed of some reactions in $G$ and has the same dimension as $G$}. Note that $G$ itself is also  a subnetwork of $G$.
%\end{definition}

%In this section, we utilize Algorithm~\ref{algorithm3} to enumerate all subnetworks of the maximum networks obtained in the previous section, i.e. all the two-dimensional zero-one networks with  three species. Then we remove the subnetworks that have the same form and networks that admit no positive steady states. According to \cite[Lemma 24]{JTZ}, for any two-dimensional zero-one network \( G \), let \( h \) be the steady-state system augmented by conservation laws defined as in \eqref{eq:h}. If \( G \) is a subnetwork of certain network in \( G_1 \), then for any \( \kappa \in \mathbb{R}_{>0}^m \) and for any corresponding positive steady state \( x \in \mathbb{R}_{>0}^3 \), we have  $\det(\operatorname{Jac}_h(\kappa, x)) > 0$}. Hence, all subnetworks of any network in $G_1$ are nondegenerate. Therefore, we only need to check all subnetworks of the networks in  $G_2$ and $G_3$. We carry out the following algorithm.
Algorithm~\ref{algorithm3} enumerates all $r$-dimensional zero‑one reaction networks with $s$ species that are subnetworks of the maximum networks. Then it removes the networks that have the same form or admit no positive steady states. The algorithm proceeds as follows.

\begin{newalgorithm}
\textit{Preprocessing}
\label{algorithm3}

\begin{description}[leftmargin=!, labelwidth=\widthof{\bf Step 4.}, align=left]
  \item[\bf Input:] A maximum $r$‑dimensional zero‑one reaction network with $s$ species.
  
  \item[\bf Output:] A list of $r$-dimensional zero-one reaction networks with $s$ species such that
        \begin{itemize}[leftmargin=*, itemsep=0pt, parsep=0pt]
          \item any two networks do not have the same form,
          \item each network admits a positive steady state.
        \end{itemize}
  % \item[\bf Step 1.] Invoke Algorithm~\ref{algorithm2} to obtain the list of maximum $r$-dimensional zero-one reaction networks with $s$ species.}

  \item[\bf Step 1.] For each input maximum network, enumerate all its subnetworks and collect them into a set $\mathcal{L}$.

  \item[\bf Step 2.] For each network $G$ in $\mathcal{L}$, remove the networks that have the same form with $G$ as follows.
       \begin{enumerate}
          \item[{\bf 2.1.}] For the given network \(G\) with $s$ species and $m$ reactions as defined in \eqref{eq:network}, each reaction can be regarded as a binary number 
\begin{equation}\label{eq:binary}
B_j = \mu_{1j}\mu_{2j}\dots\mu_{sj}\nu_{1j}\nu_{2j}\dots\nu_{sj} \quad \text{for } j = 1, 2,\dots, m.
\end{equation}
          We write $m$ 2$s$-bit binary numbers corresponding to this network and define a set \(\mathcal{B}_0 = \{B_1, B_2, \dots, B_m\}\).
          \item[{\bf 2.2.}] For the set $\mathcal{B}_0$ obtained in Step~2.1, perform
          \begin{enumerate}
              \item[{\bf 2.2.1.}] For each $j = 1, 2, \dots, m$, apply the same permutation to every binary number $B_j$ as follows. Take a binary number $B_j$  as defined in \eqref{eq:binary} from $\mathcal{B}_0$ and split it into two blocks: $\mu_{1j}\mu_{2j}\dots\mu_{sj}$ and $\nu_{1j}\nu_{2j}\dots\nu_{sj}$. Apply the same permutation simultaneously to both blocks. That is, let $\sigma = (1,2,\dots,s)$; if $\sigma_1 = (k_1, k_2, \dots, k_s)$ is a given permutation of $\sigma$, then the binary number after applying $\sigma_1$ becomes \[\mu_{k_1j}\mu_{k_2j}\dots\mu_{k_sj}\nu_{k_1j}\nu_{k_2j}\dots\nu_{k_sj}.\] Apply the same permutation to all \(\{B_1, B_2, \dots, B_m\}\) in $\mathcal{B}_0$. This yields a new set of $m$ 2$s$-bit binary numbers denoted by ${\mathcal{B}_{\sigma_1}}$.
Without loss of generality, assume that all the possible permutations on $(1,2,\ldots,s)$ are $\sigma_1, \sigma_2, \dots, \sigma_n$. For each $\sigma_i$, do the above process. Collect all resulting sets of $m$ 2$s$-bit binary numbers into a set $\{\mathcal{B}_{\sigma_i}\}_{i=1}^n$.
\end{enumerate}

(\textit{Note:} The total number $n$ of such permutations is analyzed in Remark~\ref{r_6}.)

          \item[{\bf 2.3.}] For each network \(\widetilde{G}\) in \(L \setminus \{G\}\), check whether \(\widetilde{G}\) has the same number of reactions as $G$. If not, we skip \(\widetilde{G}\). If so, we compute the set of its \(m\) \(2s\)-bit binary numbers. If this set is the same to any set in $\{\mathcal{B}_{\sigma_i}\}_{i=1}^n$, remove the network.
                
        \end{enumerate}

  \item[\bf Step 3.] For each network remaining after Step~2, check whether it admits any positive steady states as follows.
       \begin{enumerate}
          \item[{\bf 3.1.}] For each stoichiometric matrix \(\mathcal{N}\), check the signs of its non‑zero elements.
                \begin{enumerate}
                  \item[{\bf 3.1.1.}] If any row is found where all non-zero elements are either all positive or all negative, conclude that the network admits no positive steady states and discard the network.
                  \item[{\bf 3.1.2.}] Otherwise, check whether the equation \(\mathcal{N} v = 0\) admits a positive solution. If no positive solution exists, conclude that the network admits no positive steady state and discard it.
                \end{enumerate}
        \end{enumerate}

  \item[\bf Step 4.] Output the remaining list of networks.
\end{description}
\end{newalgorithm}

\begin{remark}\label{r_6}
    In Step~2.2, consider a $2s$-bit binary number as defined in \eqref{eq:binary}. Regard the number as $s$ pairs $(\mu_{ij}, \nu_{ij})$ for $i = 1, 2,\dots, s$. Define:

\begin{itemize}
    \item $n_{00}$ as the number of $(0,0)$ pairs,
    \item $n_{01}$ as the number of $(0,1)$ pairs,
    \item $n_{10}$ as the number of $(1,0)$ pairs,
    \item $n_{11}$ as the number of $(1,1)$ pairs.
\end{itemize}

\noindent Then, after performing the same permutation on both the first
$s$ bits and the last $s$ bits simultaneously, the number of distinct binary numbers that can be obtained is given by:

\[
\frac{s!}{n_{00}! \cdot n_{01}! \cdot n_{10}! \cdot n_{11}!}.
\]
\end{remark}

% \begin{remark}\label{r_100}
% In Step~2.2 each reaction in a zero‑one network with three species is encoded as a 6‑digit binary number and converted into a decimal number of one or two digits. Hence, to assign a unique assignment value to each network, in Step 2.3, we first sort the decimal numbers in descending order, then set the weights \(1, 100, 10000, \ldots\), and calculate the sum of the products to distinguish between networks that do not have the same form.
% \end{remark}

\section{Experiments}\label{sec4}

In Section \ref{4.1}, we enumerate all maximum two-dimensional zero-one reaction networks with three species
by Algorithm \ref{algorithm2}
using {\tt Python} within {\tt Visual Studio Code} \cite{vscode}. Then, we classify these maximum networks. In Section \ref{4.2}, we invoke Algorithm \ref{algorithm3} to enumerate all two-dimensional zero-one reaction networks with three species, and we determine whether they exhibit degeneracy using {\tt Mathematica} \cite{mathematica}, utilizing key functions including \texttt{NullSpace} for computing the null space, \texttt{FindInstance} for verifying non-negative solutions and \texttt{Det} for evaluating Jacobian determinants. As a result, we obtain all the degenerate networks, and we put the computational results online (\url{https://github.com/zjdong-sudo/computational-results/blob/main/computational-results.txt}). 
% By using {\tt Mathematica} to check} the system augmented with conservation laws $h$ for each degenerate network, we find that the corresponding steady-state system is equivalent to a binomial system, as shown in Theorem \ref{eq:thm1}. 
We perform all the experiments by a 2.70 GHz Intel Core i5-11400H processor (16GB total memory) under Windows 11.

\subsection{Classifying  Maximum Networks}\label{4.1}
% \begin{definition}\cite[Definition 2]{JTZ}}
% \label{def:maximum network}
% Let $G$ be an $r$-dimensional zero-one reaction network with $s$ species. We say that $G$ is a \defword{maximum $s$-species network} if for all} additional zero-one reaction involving only the species $X_1, \dots, X_s$, the dimension of the resulting network increases strictly to $r+1$.
% \end{definition}

% \begin{definition}\cite[Definition 1]{JiaoTang2025}}
% \label{def:the same form networks 1}
% We say that a network $G'$ \textit{has the same form with} another network $G$ if $G'$ can be derived from $G$ through a relabeling of the species $X_1, \ldots, X_s$ as $X'_1, \ldots, X'_s$, or a relabeling of the reactions $\mathcal{R}_1, \ldots, \mathcal{R}_m$ as $\mathcal{R}'_1, \ldots, \mathcal{R}'_m$.
% \end{definition}

In this section, we carry out the following computations. 
\begin{enumerate}
\item[{\bf (First).}]
We invoke Algorithm  \textit{EnumeratingMaximumNetworks}(3, 2) to enumerate all two-dimensional maximum three-species zero-one networks. Notice that after executing Step 5 of Algorithm \ref{algorithm2}, we obtain a total of 25 maximum networks. Subsequently, after removing the networks that have the same form according to Step 6, we obtain 8 maximum networks.
\item[{\bf (Second).}]
Consider a two-dimensional three-species zero-one network $G$, the conservation law can be written as
\[
x_1 = a x_2 + b x_3 + c,
\]
where \( a,b,c \in \mathbb{R} \).
We classify the $8$ maximum  networks into three classes $G_1$, $G_2$ and $G_3$ according to \cite[Lemma 17]{JTZ} as follows
\begin{align}
G_1 &:= \{G \mid (a,b) = \left(\frac{1}{2}, \frac{1}{2}\right), G \in \mathcal{G}\}, \label{eq:g1}\\
G_2 &:= \{G \mid (a,b) \in \{(1,0), (0,1), (0,0)\}, G \in \mathcal{G}\}, \label{eq:g2}\\
G_3 &:= \mathcal{G} \setminus \{G_1 \cup G_2\}, \label{eq:g3}
\end{align}
 where $\mathcal{G}$ denotes the set of all two-dimensional maximum three-species zero-one networks. We present all the maximum networks as follows. The computational results are consistent with those shown in \cite{JTZ}.

\subparagraph*{The set $G_1$ consists of the following  network \eqref{eq:g_9}.}

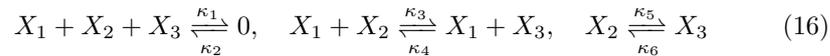
\begin{align}\label{eq:g_9}   
X_1+X_2+X_3\xrightleftharpoons[\kappa_2]{\kappa_1}0, \quad X_1+X_2\xrightleftharpoons[\kappa_4]{\kappa_3}X_1+X_3,
     \quad X_2\xrightleftharpoons[\kappa_6]{\kappa_{5}}X_3
\end{align}

\subparagraph*{The set $G_2$ consists of the following  networks \eqref{eq:g_10}--\eqref{eq:g_11}.}

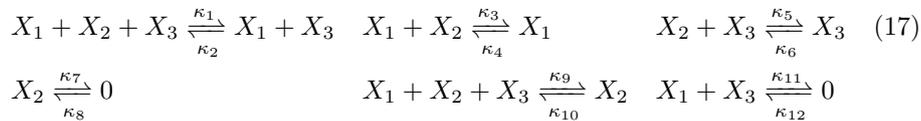
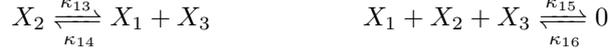
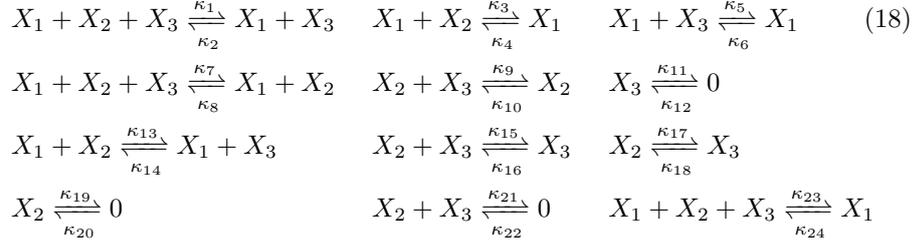
\begin{align}\label{eq:g_10}
&X_1+X_2+X_3\xrightleftharpoons[\kappa_2]{\kappa_1}X_1+X_3  
&&X_1+X_2\xrightleftharpoons[\kappa_4]{\kappa_3} X_1 
&&X_2+X_3\xrightleftharpoons[\kappa_{6}]{\kappa_{5}} X_3\\    &X_2\xrightleftharpoons[\kappa_{8}]{\kappa_{7}} 0 
&&X_1+X_2+X_3\xrightleftharpoons[\kappa_{10}]{\kappa_9}X_2  
&&X_1+X_3\xrightleftharpoons[\kappa_{12}]{\kappa_{11}}0 \notag\\
&X_2\xrightleftharpoons[\kappa_{14}]{\kappa_{13}} X_1+X_3 
&&X_1+X_2+X_3\xrightleftharpoons[\kappa_{16}]{\kappa_{15}}0\notag
\end{align}
\begin{align}\label{eq:g_11}
&X_1+X_2+X_3\xrightleftharpoons[\kappa_{2}]{\kappa_{1}}X_1+X_3
&&X_1+X_2\xrightleftharpoons[\kappa_4]{\kappa_3}X_1
&&X_1+X_3\xrightleftharpoons[\kappa_{6}]{\kappa_{5}}X_1 \\       &X_1+X_2+X_3\xrightleftharpoons[\kappa_{8}]{\kappa_{7}}X_1+X_2
&&X_2+X_3\xrightleftharpoons[\kappa_{10}]{\kappa_9}X_2       
&&X_3\xrightleftharpoons[\kappa_{12}]{\kappa_{11}}0 \notag \\
&X_1+X_2\xrightleftharpoons[\kappa_{14}]{\kappa_{13}}X_1+X_3
&&X_2+X_3\xrightleftharpoons[\kappa_{16}]{\kappa_{15}}X_3
&&X_2\xrightleftharpoons[\kappa_{18}]{\kappa_{17}}X_3\notag\\   
&X_2\xrightleftharpoons[\kappa_{20}]{\kappa_{19}}0
&&X_2+X_3\xrightleftharpoons[\kappa_{22}]{\kappa_{21}}0
&&X_1+X_2+X_3\xrightleftharpoons[\kappa_{24}]{\kappa_{23}}X_1\notag
\end{align}
\subparagraph*{The set $G_3$ consists of the following networks \eqref{eq:g_12}--\eqref{eq:g_16}.}

\begin{align}
&X_1+X_2\xrightleftharpoons[\kappa_2]{\kappa_1}X_1+X_3
&&X_1+X_3\xrightleftharpoons[\kappa_{4}]{\kappa_3}X_2+X_3
&&X_1+X_2\xrightleftharpoons[\kappa_{6}]{\kappa_{5}}X_2+X_3\label{eq:g_12}\\
&X_2\xrightleftharpoons[\kappa_8]{\kappa_7}X_3
&&X_1\xrightleftharpoons[\kappa_{10}]{\kappa_9}X_2 
&&X_1\xrightleftharpoons[\kappa_{12}]{\kappa_{11}}X_3\notag
\end{align}
\begin{align}\label{eq:g_13} 
&X_1+X_2\xrightleftharpoons[\kappa_2]{\kappa_1} X_1+X_3        &&X_1\xrightleftharpoons[\kappa_4]{\kappa_3} X_2+X_3
&&X_2\xrightleftharpoons[\kappa_6]{\kappa_5}X_3
\end{align}
\begin{align}\label{eq:g_14} 
&X_1\xrightleftharpoons[\kappa_2]{\kappa_1}X_2+X_3
&&X_1+X_2\xrightleftharpoons[\kappa_4]{\kappa_3}X_3
&&X_1\xrightleftharpoons[\kappa_{6}]{\kappa_{5}}X_3\\
&X_1+X_2\xrightleftharpoons[\kappa_8]{\kappa_7}X_3+X_2 
&&0\xrightleftharpoons[\kappa_{10}]{\kappa_{9}}X_2
&&X_1\xrightleftharpoons[\kappa_{12}]{\kappa_{11}}X_2+X_1\notag \\
&X_3\xrightleftharpoons[\kappa_{14}]{\kappa_{13}}X_2+X_3
&&X_1+X_3\xrightleftharpoons[\kappa_{16}]{\kappa_{15}}X_2+X_1+X_3\notag
\end{align}
\begin{align}\label{eq:g_15}
&X_3\xrightleftharpoons[\kappa_2]{\kappa_1}X_1+X_2 
&&X_1\xrightleftharpoons[\kappa_4]{\kappa_3} X_2+X_3+X_1
&&0\xrightleftharpoons[\kappa_6]{\kappa_5} X_2+X_3
\end{align}
\begin{align}\label{eq:g_16}
&X_3\xrightleftharpoons[\kappa_2]{\kappa_1}X_1+X_2+X_3 
&&X_1+X_2\xrightleftharpoons[\kappa_6]{\kappa_5} X_3+X_2
&&X_1\xrightleftharpoons[\kappa_{10}]{\kappa_{9}} X_2+X_3+X_1\\ 
&0\xrightleftharpoons[\kappa_4]{\kappa_3}X_1+X_2
&&X_1\xrightleftharpoons[\kappa_8]{\kappa_7} X_3
&&0\xrightleftharpoons[\kappa_{12}]{\kappa_{11}} X_2+X_3\notag       
\end{align}
\end{enumerate}

\subsection{Determining degeneracy}\label{4.2}

% \begin{definition}\cite[Definition 2.1]{joshi2017small}}
% \label{eq:subnetwork}
% For any network $G$, a \textit{subnetwork} of $G$ is a network composed of some reactions in $G$ and has the same dimension as $G$}. Note that $G$ itself is also  a subnetwork of $G$.
% \end{definition}

In this section, we find out all the degenerate two-dimensional zero-one networks with three species as follows.
\begin{enumerate}
    \item[{\bf (First).}]
    According to \cite[Lemma 24]{JTZ}, for any two-dimensional zero-one network \( G \), let \( h \) be the steady-state system augmented by conservation laws defined as in \eqref{eq:h}. If \( G \) is a subnetwork of a  certain network in \( G_1 \), then for any \( \kappa \in \mathbb{R}_{>0}^m \) and for any corresponding positive steady state \( x \in \mathbb{R}_{>0}^3 \), we have  $\det(\operatorname{Jac}_h(\kappa, x)) > 0$. Hence, all subnetworks of any network in $G_1$ are nondegenerate. Therefore, we only need to check all subnetworks of the networks in  $G_2$ and $G_3$.
     For each maximum network in $G_2$ and $G_3$, we execute Algorithm \ref{algorithm3}.
     We present the number of networks after carrying out Step 1, Step 2 and Step 3 of Algorithm \ref{algorithm3} in Table \ref{tab:extended_table}, which shows the majority of networks are excluded after checking the natural equivalence and the consistency as the total number of networks drops from 16915383 to 823310. 

\begin{table}
\refstepcounter{table}
\begin{center}
{\small {\bf Table \thetable}\ \ The Number of Networks Admitting Positive Steady States\label{tab:extended_table}}

\vskip 1mm

{\small
\begin{tabular}{l *{7}{c}|c}
\toprule
\textbf{Maximum Networks} & 
\textbf{\eqref{eq:g_10}} & 
\textbf{\eqref{eq:g_11}} & 
\textbf{\eqref{eq:g_12}} & 
\textbf{\eqref{eq:g_13}} & 
\textbf{\eqref{eq:g_14}} & 
\textbf{\eqref{eq:g_15}} & 
\textbf{\eqref{eq:g_16}} & \textbf{Total}\\
\midrule
\textbf{After Step 1} & 65259  & 16776675  & 4050  & 45  & 65259  & 45  & 4050 & 16915383\\
\textbf{After Step 2} & 40779  & 840262  & 710 & 27 & 33108 & 45 & 2055 & 916986\\
\textbf{After Step 3} & 34831 & 757989 & 518 & 9 & 28445 & 15 & 1503 & 823310\\
\bottomrule
\end{tabular}}
\end{center}

\par
\begin{flushleft}
\textbf{Note.} The first row presents the labels of the maximum networks. The last column gives the total number of each row.
\end{flushleft}
\end{table}

    \item[{\bf (Second).}]
   For each subnetwork found in the previous step, we apply Algorithm~\ref{algorithm} to determine the degeneracy.
 %see Example \ref{exalgorithm}.
   The numbers of degenerate subnetworks contained in each maximum network are presented in Table \ref{tab:degenerate}, implying that only a small fraction of networks exhibit degeneracy, totaling $3152$. Notice that $3132$ of the $3152$ degenerate networks are subnetworks of the maximum networks in $G_2$, and especially, $3096$ of them are subnetworks of \eqref{eq:g_11}. Also notice that 
   only $5$ of the  
   $8$ maximum networks contain degenerate subnetworks, which are the networks \eqref{eq:g_10}--\eqref{eq:g_12}, \eqref{eq:g_14}, and \eqref{eq:g_16}. Since the maximum networks are classified as $G_2$ and $G_3$ according to the conservation laws, the computational results reflect how conservation laws might affect the degeneracy. For instance, by
   $\eqref{eq:g2}$, we know that all subnetworks from $G_2$ have one of the three conservation laws: $x_3=x_1$, $x_3=x_2$, $x_3=c$. That means for most degenerate networks, the third species only depends on at most one of the two species $x_1$ and $x_2$. In other words, if a species depends on both the other two species, then it is highly possible that it is nondegenerate.

Table~\ref{tab:timing-alg1} reports the average running time of a single execution of Algorithm~\ref{algorithm} on two-dimensional zero-one reaction networks with three species and varying numbers of reactions.
Since we focus on two-dimensional networks, a necessary condition for admitting a positive steady state is that the network must contain at least three reactions. Moreover, as the largest networks in classes \( G_2 \) and \( G_3 \) contain 24 reactions and are known from \cite[Remark 9]{JTZ} to be nondegenerate whenever they admit a positive steady state, the degenerate (or potentially degenerate) networks we need to examine lie in the range of three to twenty-three reactions. Notice that all systems contain $3$ variables with degree at most $3$ as each network has three species and is zero-one. The number of reactions reflects the number of monomials in these polynomials. 

\end{enumerate}

% \textcolor{green}{   
% \begin{example}{2}\label{exalgorithm}
% Consider the following steady-state system and the stoichiometric matrix. 
% $$
% \left\{
% \begin{aligned}
%     f_1 &= \kappa_1 x_3 + \kappa_2 x_2 + \kappa_3 x_2 x_3 \\
%     f_2 &= \kappa_1 x_3 - \kappa_2 x_2 - \kappa_3 x_2 x_3 \\
%     f_3 &= - \kappa_1 x_3
% \end{aligned}
% \right.
% $$
% $$
% \mathcal{N} = 
% \begin{pmatrix}
%     1 &  1 &  1 \\
%     1 & -1 & -1 \\
%    -1 &  0 &  0 \\
% \end{pmatrix}
% $$
% \noindent It is obvious that the first row of stoichiometric matrix $\mathcal{N}$ maintains uniform sign. Therefore, the above network has no positive solutions.
% \end{example}
% }

\begin{table}
\refstepcounter{table}
\begin{center}
{\small {\bf Table \thetable}\ \ The Number of Degenerate Networks\label{tab:degenerate}}

\vskip 1mm

{\small
\begin{tabular}{|c|c|c|c|c|c|c|c|c|}
\hline
 & $G_1$ & \multicolumn{2}{c|}{$G_2$} & \multicolumn{5}{c|}{$G_3$} \\
\hline
\textbf{Maximum Networks} & \textbf{\eqref{eq:g_9}} & \textbf{\eqref{eq:g_10}} & \textbf{\eqref{eq:g_11}} & \textbf{\eqref{eq:g_12}} & \textbf{\eqref{eq:g_13}} & \textbf{\eqref{eq:g_14}} & \textbf{\eqref{eq:g_15}} & \textbf{\eqref{eq:g_16}} \\
\hline
\textbf{Degenerate Networks} & 0 & 36 & 3096 & 1 & 0 & 18 & 0 & 1 \\
\hline
\textbf{Total Degenerate Networks} & 0 & \multicolumn{2}{c|}{3132} & \multicolumn{5}{c|}{20} \\
\hline
\end{tabular}}
\end{center}
\end{table}

\clearpage

\begin{table}
\refstepcounter{table}
\begin{center}
{\small {\bf Table \thetable}\ \ Average Running Time of Algorithm~\ref{algorithm} for a Single Network\label{tab:timing-alg1}}

\vskip 1mm

{\small
\begin{tabular*}{\textwidth}{@{\extracolsep{\fill}}cccc}
\toprule
\multicolumn{2}{c}{\textbf{Number of reactions $\leq 13$}} &
\multicolumn{2}{c}{\textbf{Number of reactions $\geq 14$}} \\
\cmidrule(lr){1-2} \cmidrule(lr){3-4}
\textbf{Reactions} & \textbf{Time (s)} & \textbf{Reactions} & \textbf{Time (s)} \\
\midrule
3  & 0.0103 & 14 & 0.1052 \\
4  & 0.0106 & 15 & 0.3281 \\
5  & 0.0109 & 16 & 1.1540 \\
6  & 0.0116 & 17 & 4.3376 \\
7  & 0.0118 & 18 & 16.7916 \\
8  & 0.0121 & 19 & 72.7282 \\
9  & 0.0126 & 20 & 330.922 \\
10 & 0.0138 & 21 & 1382.46 \\
11 & 0.0174 & 22 & 5803.28 \\
12 & 0.0242 & 23 & 23960.3 \\
13 & 0.0441 & \multicolumn{1}{c}{} & \multicolumn{1}{c}{} \\
\bottomrule
\end{tabular*}}
\end{center}

\par
\begin{flushleft}
\textbf{Note.} For networks with fewer than $13$ reactions, the running time remains below $0.03$ seconds and increases only mildly. However, when the number of reactions exceeds $14$, the computational cost reaches several seconds and eventually increases by several orders of magnitude.
\end{flushleft}
\end{table}

 % By checking all the  systems augmented with conservation laws of  the degenerate networks, we conclude Theorem \ref{eq:thm1}.

% \begin{definition}\cite[Section 2.1]{BBH2024}}
%     A species in a reaction network is called \defword{redundant} if in each reaction, the species either does not appear  or appears on both sides of the reaction (i.e., the species is both a reactant and a product). 
% \end{definition}

% \begin{definition}
%     Two reaction networks are said to be \defword{equivalent} if one of them can be obtained from the other one by removing all redundant species.
% \end{definition}

\section{Main Results}\label{sec6}
By using {\tt Mathematica} to check the system augmented with conservation laws $h$ for each degenerate network, we find that the corresponding steady-state system is equivalent to a binomial system, as shown in Theorem \ref{eq:thm1}.
\begin{theorem}\label{eq:thm1}
    For any degenerate two-dimensional zero-one reaction network with three species, after removing all trivial species, each polynomial $f_i$ in the system augmented with conservation laws $h$ defined in \eqref{h_system} consists of two monomials with respect to $x$. 
\end{theorem}

\begin{remark}\label{rmk:computational_observations}
    Our computational enumeration confirms Theorem~\ref{eq:thm1} and yields further observations:
    \begin{itemize}
        \item The degree of each polynomial $f_i$ is at most three, which aligns with the theoretical expectation for networks with three species.
        \item The degenerate networks range in size from three to twelve reactions. Accordingly, the corresponding polynomials $f_i$ initially contain up to twelve terms and at least three terms. After combining like terms, each polynomial reduces to exactly two terms.
    \end{itemize}
\end{remark}

\noindent To provide a more intuitive demonstration of Theorem~\ref{eq:thm1}, we present two representative degenerate networks as examples.

\begin{example}
Consider the following two-dimensional network

\begin{align*}
    0 \xrightarrow{\kappa_1} X_1 + X_2 + X_3, \quad 0 \xrightarrow{\kappa_2}X_1 + X_2,
     \quad 0\xrightarrow{\kappa_{3}}X_3, \quad X_1+X_2+X_3 \xrightarrow{\kappa_4} 0.
\end{align*}
\noindent The system $f$ is given as follows
\begin{align*}
    \left\{
    \begin{aligned}
    f_1 &= \kappa_1 + \kappa_2  -\kappa_4 x_1 x_2 x_3 \\
    f_2 &= \kappa_1 + \kappa_2  -\kappa_4 x_1 x_2 x_3 \\
    f_3 &= \kappa_1 + \kappa_3  -\kappa_4 x_1 x_2 x_3 \\
    \end{aligned}
    \right.
\end{align*}

\noindent The conservation law of this network is $x_1-x_2 = c$. Hence, we obtain the system $h$ according to \eqref{h_system} as follows
\begin{align*}
    \left\{
    \begin{aligned}
    h_1 &=  x_1 - x_2 - c\\ 
    h_2 &= f_2 = \kappa_1 + \kappa_2  -\kappa_4 x_1 x_2 x_3 \\
    h_3 &= f_3 = \kappa_1 + \kappa_3  -\kappa_4 x_1 x_2 x_3 \\
    \end{aligned}
    \right.
\end{align*}It is straightforward to check that $\det(\operatorname{Jac}_h)$ is identically zero, and hence, it is degenerate. 
It is also directly observed that each polynomial  \( f_i \) consists of a constant term and the monomial \( x_1x_2x_3 \), which is consistent with Theorem \ref{eq:thm1}. 

\end{example}

\begin{example} \label{ex:last}
Consider the following two-dimensional network

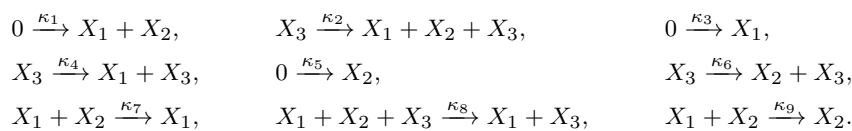
\begin{equation}\label{eq:24}
\begin{aligned}
&0 \xrightarrow{\kappa_1} X_1 + X_2, 
&&X_3 \xrightarrow{\kappa_2}X_1 + X_2+X_3, 
&& 0\xrightarrow{\kappa_{3}}X_1,\\
&X_3 \xrightarrow{\kappa_4} X_1+X_3, 
&&0 \xrightarrow{\kappa_5} X_2 , 
&&X_3 \xrightarrow{\kappa_6} X_2 + X_3,\\ 
&X_1+X_2 \xrightarrow{\kappa_7} X_1,   
&&X_1+X_2+X_3 \xrightarrow{\kappa_8} X_1  + X_3,
&&X_1+X_2\xrightarrow{\kappa_9}  X_2.
\end{aligned}
\end{equation}

\noindent The steady-state system $f$ is given as follows
\begin{align}\label{eq:ex4_f1}
    \left\{
    \begin{aligned}
    f_1 &= \kappa_1 + \kappa_2 x_3 + \kappa_3 + \kappa_4 x_3 - \kappa_9x_1x_2 \\
    f_2 &=  \kappa_1 + \kappa_2x_3 + \kappa_5 + \kappa_6x_3 - \kappa_7x_1x_2 - \kappa_8x_1x_2x_3 \\
    f_3 &= 0 \\
    \end{aligned}
    \right.
\end{align}

\noindent The conservation law of this network is $x_3= c$ where $c\in {\mathbb R}_{>0}$. Hence, we obtain the system augmented  with conservation laws $h$ according to \eqref{h_system} as follows
\begin{align}\label{eq:ex4_h1}
    \left\{
    \begin{aligned}
    h_1 &= f_1 = \kappa_1 + \kappa_2 x_3 + \kappa_3 + \kappa_4 x_3 - \kappa_9x_1x_2 \\
    h_2 &= f_2 =  \kappa_1 + \kappa_2x_3 + \kappa_5 + \kappa_6x_3 - \kappa_7x_1x_2 - \kappa_8x_1x_2x_3 \\
    h_3 &= x_3 - c \\
    \end{aligned}
    \right.
\end{align}

\noindent The Jacobian matrix of system $h$ is
\[
\operatorname{Jac}_h = 
\begin{pmatrix}
-\kappa_9 x_2 & -\kappa_9 x_1 & \kappa_2 + \kappa_4 \\[4pt]
-\kappa_7 x_2 - \kappa_8 x_2 x_3 & -\kappa_7 x_1 - \kappa_8 x_1 x_3 & \kappa_2 + \kappa_6-\kappa_8 x_1 x_2   \\[4pt]
0 & 0 & 1
\end{pmatrix}
\]
We observe that  $\det(\operatorname{Jac}_h)$ is identically zero at any steady state. Hence, the network is degenerate.

On the other hand, notice that $x_3$ is a trivial species in this network. For instance, the reactions $0 \xrightarrow{\kappa_1} X_1 + X_2$ and
$X_3 \xrightarrow{\kappa_2}X_1 + X_2+X_3$ in \eqref{eq:24}, after removing $X_3$ and renaming the rate constants, yield the reaction $0 \xrightarrow{\eta_1} X_1 + X_2$. Similarly, the reactions indexed by $\kappa_3$ and $\kappa_4$, the reactions indexed by $\kappa_5$ and $\kappa_6$, and the reactions indexed by $\kappa_7$ and $\kappa_8$ can be combined respectively. In this way, we obtain a two-species network after removing $X_3$:

\begin{equation} \label{eq:25}
\begin{aligned}
&0 \xrightarrow{\eta_1} X_1 + X_2,
&&0 \xrightarrow{\eta_2} X_1,
&&0 \xrightarrow{\eta_3} X_2,\\
&X_1 + X_2 \xrightarrow{\eta_4} X_1,
&&X_1 + X_2 \xrightarrow{\eta_5} X_2.
\end{aligned}
\end{equation}

\noindent The new steady-state system $\tilde{f}$ is given as follows
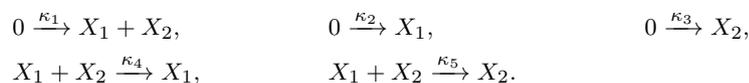
\begin{align}\label{eq:ex4_f2}
    \left\{
    \begin{aligned}
    \tilde{f}_1 &= \eta_1 + \eta_2 - \eta_5 x_1 x_2\\
    \tilde{f}_2 &= \eta_1 + \eta_3 - \eta_4 x_1 x_2 \\
    \end{aligned}
    \right.
\end{align}

\noindent Notice that there is no conservation law since the network is full-dimensional. So, there is no conservation laws any longer, and the system augmented with conservation laws $\tilde{h}$ is equal to $\tilde{f}$. %The Jacobian matrix of system $\tilde{h}$ is
%\[
%\operatorname{Jac}_{\tilde{h}} = 
%\begin{pmatrix}
%-\eta_5 x_2 & -\eta_5 x_1 \\[4pt]
%-\eta_4 x_2 & -\eta_4 x_1  \\%[4pt]
%\end{pmatrix}
%\]
It is easy to see that $\det(\operatorname{Jac}_{\tilde{h}})$ is identically zero. Hence, the network is degenerate. In particular, each polynomial $\tilde{f}_i$ consists of a constant term and the monomial $x_1x_2$, which verifies the conclusion of Theorem \ref{eq:thm1}. 

By this example we can see why the network \eqref{eq:24} is equivalent to \eqref{eq:25}. 
For the system $h$ \eqref{eq:ex4_h1}, by setting $h_3 = 0$ and substituting $x_3 = c$ into $h_1$ and $h_2$, we have  
\begin{align}
    \left\{
    \begin{aligned}
    & \kappa_1 + \kappa_2 c + \kappa_3 + \kappa_4 c - \kappa_9x_1x_2 \\
    & \kappa_1 + \kappa_2 c + \kappa_5 + \kappa_6 c - \kappa_7x_1x_2 - \kappa_8cx_1x_2, 
    \end{aligned}
    \right.
\end{align}
which can be considered as the same system with the system $\tilde{f}$ (and also $\tilde{h}$) in \eqref{eq:ex4_f2} by setting $\kappa_1 + \kappa_2 c + \kappa_3 + \kappa_4 c = \eta_1+\eta_2$, $\kappa_1 + \kappa_2 c + \kappa_5 + \kappa_6 c = \eta_1+\eta_3$,
$\kappa_9=\eta_5$, and $\kappa_7+\kappa_8=\eta_4$.

\end{example}

%\begin{blackremark}{\textbf{6}}
%Remark that the above example also shows that  any  degenerate two-dimensional zero-one network with three species that contains a trivial} species is equivalent to a degenerate  two-dimensional zero-one network with two species.
%\end{blackremark}

\section{Discussion}\label{sec7}
In this work, we identify all the degenerate two-dimensional zero-one networks with three species by a standard method, and by studying their systems augmented with conservation laws, we obtain Theorem~\ref{eq:thm1}.  In the future, a productive direction is to seek a mathematical proof for Theorem \ref{eq:thm1}. We believe that the proof will  theoretically explain why most degenerate networks are from the group $G_2$ and provide explicit characterizations of degenerate two-dimensional zero-one networks.
Also, it would be interesting to 
explore whether Theorem~\ref{eq:thm1} can be generalized to higher-dimensional networks.

\section*{Conflict of Interest}
The authors declare no conflict of interest.

% \section*{Acknowledgements}
% Please add acknowledgements here if needed.


\begin{thebibliography}{99}

%\bibitem{Blanco2013}
%Celia Blanco, Joaquim Crusats, Zoubir El-Hachemi, Albert Moyano, David Hochberg, and Josep M. Rib\'o.
%\newblock Spontaneous emergence of chirality in the limited enantioselectivity model: autocatalytic cycle driven by an external reagent.
%\newblock \textit{ChemPhysChem}, 14:2432–2440, 2013.
%\newblock \href{https://doi.org/10.1002/cphc.201300350}{https://doi.org/10.1002/cphc.201300350}}

\bibitem{conradi2019existence}
Carsten Conradi, Elisenda Feliu, and Maya Mincheva.
\newblock On the existence of Hopf bifurcations in the sequential and distributive double phosphorylation cycle.
\newblock \textit{Mathematical Biosciences and Engineering}, 17(1):494–513, 2019.
\newblock \href{https://doi.org/10.3934/mbe.2020027}{https://doi.org/10.3934/mbe.2020027}

\bibitem{2208.04196}
Xiaoxian Tang, and Kaizhang Wang.
\newblock Hopf bifurcations of reaction networks with zero-one stoichiometric coefficients.
\newblock \textit{SIAM Journal on Applied Dynamical Systems}, 22(3):2459–2489, 2023.
\newblock \href{https://doi.org/10.1137/22m1519754}{https://doi.org/10.1137/22m1519754}

\bibitem{BanajiBoros}
Murad Banaji, and Balázs Boros.
\newblock The smallest bimolecular mass action reaction networks admitting Andronov–Hopf bifurcation.
\newblock \textit{Nonlinearity}, 36(2):1398, 2023.
\newblock \href{https://doi.org/10.1016/j.aml.2023.108671}{https://doi.org/10.1016/j.aml.2023.108671}

\bibitem{Bihan2020}
Frédéric Bihan, Alicia Dickenstein, and Magalí Giaroli.
\newblock Lower bounds for positive roots and regions of multistationarity in chemical reaction networks.
\newblock \textit{Journal of Algebra}, 542:367–411, 2020.
\newblock \href{https://doi.org/10.1016/j.jalgebra.2019.10.002}{https://doi.org/10.1016/j.jalgebra.2019.10.002}

\bibitem{Mueller2016}
Stefan Müller, Elisenda Feliu, Georg Regensburger, Carsten Conradi, Anne Shiu, and Alicia Dickenstein.
\newblock Sign conditions for injectivity of generalized polynomial maps with applications to chemical reaction networks and real algebraic geometry.
\newblock \textit{Foundations of Computational Mathematics}, 16(1):69–97, 2016.
\newblock \href{https://doi.org/10.1007/s10208-0149239-3}{https://doi.org/10.1007/s10208-0149239-3}

\bibitem{Puente2025}
Luis David García Puente, Elizabeth Gross, Heather A. Harrington, Matthew Johnston, Nicolette Meshkat, Mercedes Pérez Millán, and Anne Shiu.
\newblock Absolute concentration robustness: Algebra and geometry.
\newblock \textit{Journal of Symbolic Computation}, 128:102398, 2025.
\newblock \href{https://doi.org/10.1016/j.jsc.2024.102398}{https://doi.org/10.1016/j.jsc.2024.102398}

\bibitem{Feliu2024}
Elisenda Feliu, Oskar Henriksson, and Beatriz Pascual‑Escudero.
\newblock The generic geometry of steady state varieties.
\newblock  arXiv:2412.17798, 2024.
\newblock \href{https://arxiv.org/abs/2412.17798}{https://arxiv.org/abs/2412.17798}

\bibitem{CCW}
Changbo Chen and Wenyuan Wu.
\newblock A geometric approach for analyzing parametric biological systems by exploiting block triangular structure.
\newblock \textit{SIAM Journal on Applied Dynamical Systems}, 21(2):1573--1596, 2022.
\newblock \href{https://doi.org/10.1137/21M1436373}{https://doi.org/10.1137/21M1436373}

\bibitem{MJ}
Chenqi Mou and Wenwen Ju.
\newblock Sparse triangular decomposition for computing equilibria of biological dynamic systems based on chordal graphs.
\newblock \textit{IEEE/ACM Transactions on Computational Biology and Bioinformatics}, 20(3):1667--1678, 2023.
\newblock \href{https://doi.org/10.1109/TCBB.2022.3156759}{https://doi.org/10.1109/TCBB.2022.3156759}

\bibitem{BanajiPantea2018}
Murad Banaji, and Casian Pantea.
\newblock The inheritance of nondegenerate multistationarity in chemical reaction networks.
\newblock \textit{SIAM Journal on Applied Mathematics}, 78(2):1105–1130, 2018.
\newblock \href{https://doi.org/10.1137/16m1103506}{https://doi.org/10.1137/16m1103506}

\bibitem{JoshiShiu2013}
Badal Joshi, and Anne Shiu.
\newblock Atoms of multistationarity in chemical reaction networks.
\newblock \textit{Journal of Mathematical Chemistry}, 51:153–178, 2013.
\newblock \href{https://doi.org/10.1007/s10910-012-0072-0}{https://doi.org/10.1007/s10910-012-0072-0}

\bibitem{Banaji2018}
Murad Banaji.
\newblock Inheritance of oscillation in chemical reaction networks.
\newblock \textit{Appl. Math. Comput.}, 325:191–209, 2018.
\newblock 
\href{https://doi.org/10.1016/j.amc.2017.12.012}{https://doi.org/10.1016/j.amc.2017.12.012}

\bibitem{BanajiBorosHofbauer2023}
Murad Banaji, Balázs Boros, and Josef Hofbauer.
\newblock The inheritance of local bifurcations in mass action networks.
\newblock  arXiv:2312.12897, 2023.
\newblock \href{https://arxiv.org/abs/2312.12897}{https://arxiv.org/abs/2312.12897}

\bibitem{Ramakrishnan2008}
Naren Ramakrishnan, and Upinder S. Bhalla.
\newblock Memory switches in chemical reaction space.
\newblock \textit{PLoS Computational Biology}, 4(7):e1000122, 2008.
\newblock \href{https://doi.org/10.1371/journal.pcbi.1000122}{https://doi.org/10.1371/journal.pcbi.1000122}

\bibitem{Wilhelm2009}
Thomas Wilhelm.
\newblock The smallest chemical reaction system with bistability.
\newblock \textit{BMC Systems Biology}, 3:90, 2009.
\newblock \href{https://doi.org/10.1186/1752-0509-3-90}{https://doi.org/10.1186/1752-0509-3-90}


\bibitem{TangXu2021}
Xiaoxian Tang, and Hao Xu.
\newblock Multistability of small reaction networks.
\newblock \textit{SIAM Journal on Applied Dynamical Systems}, 20(2):608–635, 2021.
\newblock \href{https://doi.org/10.1137/20m1358761}{https://doi.org/10.1137/20m1358761}

\bibitem{KaihnsaNguyenShiu2024}
Nidhi Kaihnsa, Tung Nguyen, and Anne Shiu.
\newblock Absolute concentration robustness and multistationarity in reaction networks: Conditions for coexistence.
\newblock \textit{European Journal of Applied Mathematics}, 35(4):566--600, 2024.
\newblock 
\href{https://doi.org/10.1017/s0956792523000335}{https://doi.org/10.1017/s0956792523000335}

\bibitem{JiaoTang2025}
Yue Jiao, and Xiaoxian Tang.
\newblock An Efficient Algorithm for Determining the Equivalence of Zero-one Reaction Networks.
\newblock \textit{International Symposium on Symbolic and Algebraic Computation (ISSAC '25)}, 277--283, 2025.
\newblock \href{https://doi.org/10.1145/3747199.3747571}{https://doi.org/10.1145/3747199.3747571}

\bibitem{Dickenstein2023}
Alicia Dickenstein, Magalí Giaroli, Mercedes Pérez Millán, and Rick Rischter.
\newblock Multistationarity questions in reduced vs extended biochemical networks.
\newblock  arXiv:2310.02455, 2023.
\newblock \href{https://arxiv.org/abs/2310.02455}{https://arxiv.org/abs/2310.02455}

\bibitem{Zumsande2010}
Martin Zumsande, and Thilo Gross.
\newblock Bifurcations and chaos in the MAPK signaling cascade.
\newblock \textit{Journal of Theoretical Biology}, 265(3):481--491, 2010.
\newblock \href{https://doi.org/10.1016/j.jtbi.2010.05.036}{https://doi.org/10.1016/j.jtbi.2010.05.036}

\bibitem{JaniakSpens2005}
Fabiola Janiak-Spens, Paul F. Cook, and Ann H. West.
\newblock Kinetic analysis of YPD1-dependent phosphotransfer reactions in the yeast osmoregulatory phosphorelay system.
\newblock \textit{Biochemistry}, 44(1):377–386, 2005.
\newblock \href{https://doi.org/10.1021/bi048433s}{https://doi.org/10.1021/bi048433s}

\bibitem{Kothamachu2015}
Varun B. Kothamachu, Elisenda Feliu, Luca Cardelli and Orkun S. Soyer.
\newblock Unlimited multistability and Boolean logic in microbial signalling.
\newblock \textit{Journal of the Royal Society Interface}, 12(108):20150234, 2015.
\newblock \href{https://doi.org/10.1098/rsif.2015.0234}{https://doi.org/10.1098/rsif.2015.0234}

\bibitem{Georgiev2006}
Nikola Georgiev, Valko Petrov, Elena Nikolova, and Georgi Georgiev.
\newblock Qualitative Modelling of Quasi-homogeneous Effects in ERK and STAT Interaction Dynamics.
\newblock \textit{International Journal Bioautomation}, 5:78, 2006.

\bibitem{Baudier2018}
Adrien Baudier, François Fages, and Sylvain Soliman.
\newblock Graphical requirements for multistationarity in reaction networks and their verification in BioModels.
\newblock \textit{Journal of Theoretical Biology}, 459:79--89, 2018.
\newblock \href{https://doi.org/10.1016/j.jtbi.2018.09.010}{https://doi.org/10.1016/j.jtbi.2018.09.010}

\bibitem{TelekFeliu2023}
Telek, M{\'a}t{\'e} L{\'a}szl{\'o} and Feliu, Elisenda.
\newblock Topological descriptors of the parameter region of multistationarity: Deciding upon connectivity.
\newblock \textit{PLoS Computational Biology}, 19(3):e1010970, 2023.
\newblock \href{https://doi.org/10.1371/journal.pcbi.1010970}{https://doi.org/10.1371/journal.pcbi.1010970}

\bibitem{JTZ}
Yue Jiao, Xiaoxian Tang, and Xiaowei Zeng.
\newblock Multistability of small zero--one reaction networks.
\newblock \textit{Journal of Mathematical Biology}, 91:82, 2025.
\newblock \href{https://doi.org/10.1007/s00285-025-02306-w}{https://doi.org/10.1007/s00285-025-02306-w}

{\bibitem{PerezMillan2012}
Mercedes P\'erez Mill\'an, Alicia Dickenstein, Anne Shiu, and Carsten Conradi.
\newblock Chemical reaction systems with toric steady states.
\newblock \textit{Bulletin of Mathematical Biology}, 74:1027–1065, 2012.
\newblock \href{https://doi.org/10.1007/s11538-011-9685-x}{https://doi.org/10.1007/s11538-011-9685-x}}

\bibitem{Sadeghimanesh2019}
A.~Sadeghimanesh and E.~Feliu.
\newblock The multistationarity structure of networks with intermediates and a binomial core network.
\newblock \textit{Bulletin of Mathematical Biology}, 81(6):2428--2462, 2019.
\newblock \url{https://doi.org/10.1007/s11538-019-00612-1}

\bibitem{Rahkooy2021}
Hamid Rahkooy, and Thomas Sturm.
\newblock Testing binomiality of chemical reaction networks using comprehensive Gröbner systems.
\newblock \textit{Lecture Notes in Computer Science}, 12865:334–352, 2021.
\newblock 
\url{https://doi.org/10.1007/978-3-030-85165-1_19}

\bibitem{ConradiPantea_multistationarity}
Carsten Conradi, and Casian Pantea.
\newblock Multistationarity in biochemical networks: results, analysis, and examples.
\newblock In \textit{Algebraic and combinatorial computational biology}, pages 279--317, Elsevier, 2019.

\bibitem{joshi2017small}
Badal Joshi, and Anne Shiu.
\newblock Which small reaction networks are multistationary?
\newblock \textit{SIAM Journal on Applied Dynamical Systems}, 16(2):807–869, 2017.
\newblock \href{https://doi.org/10.1137/16M1069705}{https://doi.org/10.1137/16M1069705}

\bibitem{BBH2024}
Murad Banaji, Bal\'{a}zs Boros, and Josef Hofbauer.
\newblock Oscillations in three-reaction quadratic mass-action systems.
\newblock \textit{Studies in Applied Mathematics}, 152:249--278, 2024.
\newblock \href{https://doi.org/10.1111/sapm.12639}{https://doi.org/10.1111/sapm.12639}.

\bibitem{WiufFeliu_powerlaw}
Carsten Wiuf, and Elisenda Feliu.
\newblock Power-law kinetics and determinant criteria for the preclusion of multistationarity in networks of interacting species.
\newblock \textit{SIAM J Appl Dyn Syst}, 12:1685--1721. \href{https://doi.org/10.1137/120873388}{https://doi.org/10.1137/120873388}


\bibitem{vscode}
Microsoft Corporation.
\newblock Visual Studio Code.
\newblock Version 1.100, 2025. 

\bibitem{mathematica}
Wolfram Research, Inc.
\newblock Mathematica.
\newblock Version 14.2, Champaign, IL, 2024





\end{thebibliography}
\end{document}